\documentclass[notitlepage,floatfix,aps,pra,superscriptaddress,twocolumn,footinbib]{revtex4-2}
\usepackage{amsmath,amsfonts,amssymb,amscd,mathtools}
\usepackage{amsthm}
\usepackage{mathrsfs}
\usepackage{dsfont}
\theoremstyle{plain}
\newtheorem{theorem}{Theorem}

\newtheorem{lemma}{Lemma}

\theoremstyle{definition}
\newtheorem{definition}{Definition}
\theoremstyle{remark}

\usepackage{enumerate}
\usepackage{braket}
\usepackage{bm}
\usepackage{graphicx}
\usepackage{verbatim}
\usepackage{hyperref}
\usepackage[dvipsnames]{xcolor}
\hypersetup{
    pdfnewwindow=true,
    colorlinks=true,
    citecolor=Blue,
    linkcolor=Blue,
    urlcolor=Blue
}
\usepackage{tikz}
\DeclareMathOperator{\Tr}{Tr}
\DeclareMathAlphabet{\mathpzc}{OT1}{pzc}{m}{it}
\DeclareMathAlphabet{\mathcalligra}{T1}{calligra}{m}{n}

\newcommand{\cE}{\mathcal{E}}

\newcommand{\cG}{\mathcal{G}}
\newcommand{\cH}{\mathcal{H}}

\newcommand{\cM}{\mathcal{M}}

\newcommand{\cS}{\mathcal{S}}

\newcommand{\cX}{\mathcal{X}}
\newcommand{\cY}{\mathcal{Y}}

\newcommand{\rA}{\mathrm{A}}
\newcommand{\rB}{\mathrm{B}}
\newcommand{\rC}{\mathrm{C}}

\newcommand{\rG}{\mathrm{G}}
\newcommand{\rH}{\mathrm{H}}

\newcommand{\ketbra}[2]{\ket{#1}\!\bra{#2}}

\begin{document}

\title{One-shot and asymptotic classical capacity in general physical theories}
\author{Shintaro Minagawa}
\email{minagawa.shintaro@nagoya-u.jp}
\affiliation{Graduate School of Informatics, Nagoya University, Furo-cho, Chikusa-ku, Nagoya 464-8601, Japan}
\author{Hayato Arai}
\email{hayato.arai@riken.jp}
\affiliation{RIKEN, Center for Quantum Computing, Mathematical Quantum Information RIKEN Hakubi Research Team, 2-1 Hirosawa, Wako 351-0198, Japan}

\begin{abstract}
    With the recent development of quantum information theory, some attempts exist to construct information theory beyond quantum theory. Here we consider hypothesis testing relative entropy and one-shot classical capacity, that is, the optimal rate of classical information transmitted by using a single channel under a constraint of a certain error probability, in general physical theories where states and measurements are operationally defined. Then we obtain the upper bound of one-shot classical capacity by generalizing the method given by Wang and Renner [Phys. Rev. Lett. 108, 200501 (2012)].
    Also, we derive the lower bound of the capacity by showing the existence of a good code that can transmit classical information with a certain error probability.
    Applying the above two bounds, we prove the asymptotic equivalence between classical capacity and hypothesis testing relative entropy even in any general physical theory.
\end{abstract}

\maketitle

\section{Introduction}
Since Shannon invented the information theory \cite{shannon1948mathematical}, it has been increasingly important (see e.g., Ref.~\cite{Cover2006}).
The goal of information theory is basically to express the optimal efficiency for some tasks, and the optimal efficiencies for different tasks are sometimes equivalent or directly related through some information quantities like mutual information; a typical example is the asymptotic equivalence between the exponential rate of hypothesis testing and the classical information transmission capacity~\cite{verdu1994general}.

Recently, as quantum information theory (see e.g., Refs.~\cite{wilde2017quantum,watrous2018theory}) has flourished, similar relations are known in quantum theory.
In particular, the same relationship between hypothesis testing and channel capacity also holds in quantum theory~\cite{hiai1991proper,ogawa2005strong,schumacher1995quantum,schumacher1997sending,holevo1998capacity,bennett2002entanglement}.
Such facts imply that an information theory should possess such relations between the optimal efficiencies for some tasks independently of the mathematical structure of its background physical systems.

However, when we establish an information theory standing by the operationally minimum principles,
possible models of background physical systems are not restricted to classical and quantum theory.
Such theories are called \textit{General Probabilistic Theories}~\cite{janotta2014generalized,dariano-OPT,plavala2023general,muller2021probabilistic,ludwig1964,ludwig1967,davies1970operational,gudder1973,ozawa1980optimal,popescu1994quantum,hardy2001quantum,barrett2007information,pawlowski2009information,kimura2009optimal,short2010entropy,barnum2010entropy,kimura2010distinguishability,nuida2010optimal,chiribella2010probabilistic,Chiribella2011Informational,muller2012structure,barnum2014higher,chiribella2015entanglement,chiribella2015operational,kimura2016entropies,krumm2017thermodynamics,jencova2018incompatible,arai2019perfect,yoshida2020perfect,shahandeh2021contextuality,selby2021accessible,wakakuwa2021gentle,perinotti2022shannon,aubrun2021entangleability,aubrun2022entanglement,jenvcova2022assemblages,minagawa2022neumann,perinotti2023entropy}, in short, GPTs (for a review, see, e.g., Refs.~\cite{dariano-OPT,janotta2014generalized,plavala2023general,muller2021probabilistic}).
The framework of GPTs is a kind of generalization of classical and quantum theory
whose states and measurements are operationally defined,
and studies of GPTs have been widespread recently.

Even in such general models, some properties of information theory also hold similarly to quantum theory.
One of such results of preceding studies of GPTs is the no-cloning theorem in GPTs~\cite{barnum2006cloning}.
It is clarified that any model except for classical theory cannot copy any information freely, similar to the no-cloning theorem in quantum theory, which means that quantum theory is not a special theory with no-cloning, but the classical theory is a special theory with cloning.

On the other hand, some properties of information theory are drastically changed in GPTs.
A typical example is \emph{entropy}.
In quantum theory, there are several methods to characterize von Neumann entropy $S(\rho):=-\Tr[\rho\log\rho]$~\cite{von1955mathematical} based on the classical Shannon entropy, but the von Neumann entropy is well-defined without a choice of the way of classicalization~\cite{wilde2017quantum,watrous2018theory}.
On the other hand, it is known that such well-definedness is not valid in GPTs,
i.e.,
there is no simple way to define entropy similar to the von Neumann entropy in GPTs \cite{barnum2010entropy,short2010entropy,kimura2010distinguishability,perinotti2022shannon,perinotti2023entropy}.
A certain generalization of von Neumann entropy is not even concave~\cite{short2010entropy,kimura2010distinguishability}.

Because entropy is not generalized straightforwardly in GPTs, we cannot easily obtain a similar result of optimal efficiency for certain information tasks.
Therefore, it is a difficult question whether there are the same relations between optimal efficiencies for different tasks as the relations in classical and quantum theory.
If the answer is positive, i.e., relations between different information tasks are independent of the mathematical structure of physical systems even though entropies do not behave the same, we can reach a new foundational perspective on information theory;
``Efficiencies for information tasks give more robust definitions of information quantities than entropies in GPTs because of the independence of the mathematical structure of physical systems.''

In this paper, we discuss hypothesis testing and classical information transmission in GPTs in the same way as classical and quantum theory following Ref.~\cite{wang2012one}.
Next, we estimate the upper and lower bound of one-shot classical capacity by hypothesis testing relative entropy\footnote{This quantity was first introduced by Ref.~\cite{buscemi2010quantum}, and the relationship to one-shot classical capacity was derived by Ref.~\cite{wang2012one}, but we generalize it to GPTs in this paper.} in GPTs.
As a result, we obtain upper and lower bounds similar to that of quantum theory.
Moreover, due to the construction of the achievable case of our bound, our result of the one-shot case can be applied to the asymptotic case even though the asymptotic scenario is complicated in GPTs.
Consequently, we show the asymptotic equivalence between the above two efficiencies even in GPTs.

The subsequent structure of this paper is as follows.
In section \ref{section:GPTs}, we introduce basic notions of GPTs.
In section \ref{section:HTRE}, we introduce hypothesis testing relative entropy in GPTs.
In section \ref{section:capacity}, we derive the one-shot classical capacity theorem in GPTs.
In section \ref{section:asymptotic}, we introduce the asymptotic setting and show the asymptotic equivalence between on-shot classical capacity and hypothesis testing relative entropy in GPTs.

\section{General probabilistic theories}\label{section:GPTs}
Here we introduce the mathematical basics of GPTs following Refs.~\cite{yoshida2020perfect,janotta2014generalized,plavala2023general,muller2021probabilistic}.
Let $V$ be a finite-dimensional real vector space and the subset $K\subset V$ be a positive cone, i.e., a set satisfying the following three conditions: (i) $\lambda x\in K$ holds for any $x\in K$ and any $\lambda\ge0$. (ii) $K$ is convex and has a non-empty interior. (iii) $K\cap(-K)=\{0\}$.
The dual cone of $K$, denoting $K^*$ is defined as follows:
\begin{equation}
	K^*:=\{y\in V^*\mid \langle y,x\rangle\ge 0\ \forall x\in K\}
\end{equation}
where $\langle,\rangle$ is the inner product of the vector space $V$.
Besides, an inner point $u\in K^*$, called \textit{unit effect}, is fixed for a model.
Then, a state in this model is defined as an element $\rho\in K$ satisfying $\langle \rho,u\rangle=1$. The state space, i.e., the set of all states, is denoted as $S(K)$.
Due to the convexity of $K$, the state space $S(K)$ is also convex.

Also, a measurement is defined as a family $\bm{e}:=\{e_j\}_{j\in J}$ satisfying $e_j\in K^\ast$ for any $j\in J$ and $\sum_{j\in J} e_j=u$.
The measurement space, i.e., the set of all measurements with finite outcomes, is denoted as $\cM(K)$.
Here, $\langle e_j,\rho\rangle$ corresponds to the probability of obtaining an outcome $j\in J$ when we perform a measurement $\mathbf{e}$ to a state $\rho\in\cS(K)$.
Next, we define an order relation $\ge$ on $K^*$.
We say that $f\ge e$ if $f-e\in K^*$.
This means that for any element $x\in K$, $\langle f,x\rangle\ge\langle e,x\rangle$.
Here, we remark that a family $\{e,u-e\}$ is a measurement in $\cM(K)$
if and only if the element $e\in K^*$ satisfies $0\le e\le u$ by using the above order relation.

Next, we give examples of positive cones and models of GPTs.
The simplest example of positive cones is the positive part of the vector space $\mathbb{R}^n$ defined as
\begin{align}
    \mathbb{R}^n_{+}:=\{(x_i)_{i=1}^n\in\mathbb{R}^n\mid x_i\ge0 \ \forall i\}\;.
\end{align}
Considering the standard inner product on $\mathbb{R}^n$, the dual $\mathbb{R}^{n*}_{+}$ is equivalent to itself.
Fix an unit $\mathbf{u}=(1,1,\cdots,1)$.
Then, the state space $\cS(\mathbb{R}^{n*}_{+})$ is given as the set of all ensembles $\sum_{j\in J} p_j \mathbf{c}^j$,
where $\{p_j\}$ is probability vector and $\mathbf{c}^j:=(c^j_i)_{i=1}^n$ is a vector such that $c^j_i=\delta_{ij}$ with the Kronecker delta $\delta_{ij}$, which corresponds to a classical bit. In this paper, we denote $\mathbf{c}^j$ as $\ketbra{j}{j}$.
Also, a measurement is given as a family $\{\mathbf{e}^j\}_{j\in J}$ of $\mathbf{e}^j\in\mathbb{R}^{n}_{+}$ such that $\sum_j \mathbf{e}^j=\mathbf{u}$,
which corresponds to a strategy to obtain information from an ensemble of the classical $n$-level system.

This model corresponds to a theory of classical information and classical operations, i.e., classical theory.
Quantum theory is also a model of GPT in the case $V=L_{\rH}(\cH)$, $\langle x,y\rangle=\Tr xy$, $K=L_{\rH}^+(\cH)$, and $u=\mathds{1}$.
Here $L_{\rH}(\cH)$ denotes the set of Hermitian matrices on a Hilbert space $\cH$ and $L_{\rH}^+(\cH)$ denotes the set of positive semi-definite matrices on $\cH$.
Also, $\mathds{1}$ is the identity matrix on $\cH$.
In this model, a state is given as a density matrix, and a measurement is given as a POVM.

Next, we define a \textit{measurement channel} associated with a measurement $\mathbf{e}$ as the following map $\cE_{\mathbf{e}}$ from $\cS(K)$ to $\cS(\mathbb R^n_+)$~\cite{plavala2023general}:
\begin{equation}
    \mathcal E_{\mathbf{e}}(\rho):=\sum_{j\in J}\langle e_j,\rho\rangle\ketbra{j}{j}\;.
\end{equation}
We also define an adjoint map of a measurement channel $\cE_{\mathbf{e}}$ as the following map $\cE_{\mathbf{e}}^\dagger$ from $\mathbb{R}^{n*}_+=\mathbb{R}^{n}_+$ to $K^*$ for any $f\in \mathbb{R}^{n}_+$:
\begin{equation}
    \cE^\dagger_{\mathbf{e}}(f):=\sum_{j\in J}\langle f,\ketbra{j}{j}\rangle e_j\;.
\end{equation}
Note that the following equation holds for any $f\in \mathbb{R}^{n}_+$ and any $\rho\in \cS(K)$:
\begin{equation}\label{eq:adjoint}
\begin{split}
    \langle \cE^\dagger_{\mathbf{e}}(f),\rho\rangle&=\left\langle\sum_{j\in J}\langle f,\ketbra{j}{j}\rangle e_j,\rho\right\rangle\\
    &\stackrel{(a)}{=}\sum_{j\in J}\langle e_j,\rho\rangle\langle f,\ketbra{j}{j}\rangle\\
    &\stackrel{(b)}{=}\left\langle f,\sum_{j\in J}\langle e_j,\rho\rangle\ketbra{j}{j}\right\rangle=\langle f,\cE_{\mathbf{e}}(\rho)\rangle\;.
\end{split}
\end{equation}
The equality $(a)$ and $(b)$ hold because of the linearity of the inner product.
Besides, if $0\le f\le u$, it holds that $0\le \cE^\dagger_{\mathbf{e}}(f)\le u$ because $0\le \langle f, \ketbra{j}{j}\rangle\le 1$ holds.
As a result, if $f$ is an effect of a measurement $\{f,u-f\}$, then $\cE^\dagger_{\mathbf{e}}(f)$ is also an effect of the measurement $\{\cE^\dagger_{\mathbf{e}}(f),u-\cE^\dagger_{\mathbf{e}}(f)\}$.

Finally, we define a composite system of classical theory and a general model of the positive cone $K$ in GPTs~\cite{janotta2014generalized,plavala2023general}.
The vector space of the composite system is given by $\mathbb{R}^n\otimes V$.
The corresponding positive cone is given as
\begin{align}
    \mathbb{R}^n_{+}\otimes K:=\mathrm{Conv}\left(\{\ketbra{x}{x}\otimes\rho\mid x\in\mathbb{R}^n_{+}, \rho\in K\}\right),
\end{align}
where $\mathrm{Conv}(S)$ is the convex hull of the set $S$.
The unit is given as the tensor product of units in each system.
A state in the composite system is given as an ensemble of tensor products $\ketbra{x}{x}\otimes \rho$ of a classical bit $\ketbra{x}{x}$ and a general state $\rho\in\cS(K)$.
For a bipartite state $\rho^{\rA\rB}$, the marginal states $\rho^{\rA}$ and $\rho^{\rB}$ are defined as the unique states satisfying the following relations for any pair of a classical measurement $\{e_i^{\rA}\}$ and a general measurement $\{e_j^{\rB}\}$, respectively~\cite{janotta2014generalized}:
\begin{align}
    \sum_j\langle e_i^{\rA}\otimes e_j^{\rB},\rho^{\rA\rB}\rangle=\langle e_i^{\rA},\rho^{\rA}\rangle,\\
    \sum_i\langle e_i^{\rA}\otimes e_j^{\rB},\rho^{\rA\rB}\rangle=\langle e_j^{\rB},\rho^{\rB}\rangle\;.
\end{align}
Here, we note that the marginal states are given as follows if the bipartite state is given as the ensemble $\sum_{j,j'}p_{j,j'}\ketbra{j}{j}^{\rA}\otimes\rho_{j'}^{\rB}$:
\begin{align}
    \rho^{\rA}&=\sum_{j,j'}p_{j,j'}\ketbra{j}{j}^{\rA},\\
   \rho^{\rB}&=\sum_{j,j'}p_{j,j'}\rho_{j'}^{\rB}.
\end{align}

\section{Hypothesis testing relative entropy in GPTs}\label{section:HTRE}
Next, we introduce hypothesis testing relative entropy in general models.
In quantum theory, hypothesis testing relative entropy is defined for $0\le \epsilon\le 1$ as follows \cite{buscemi2010quantum,buscemi2010distilling,wang2012one,dupuis2014generalized}:
\begin{equation}
    D^\epsilon_{\rH}(\rho||\sigma):=-\log_2 \min_{\substack{E:0\le E\le \mathds{1},\\\Tr\{E\rho\}\ge 1-\epsilon}}\Tr\{E\sigma\}\;,\label{eq:def_htre}
\end{equation}
where $\mathds{1}$ is an identity operator.
This definition comes from the hypothesis testing of two quantum states (see e.g., Ref.~\cite{hayashi2017quantum}).
We discriminate the state $\rho$ and $\sigma$ by performing a two-valued measurement with a POVM $\{E,\mathds{1}-E\}$.
There are two kinds of error probabilities, type-I error probability $\Tr\{(\mathds{1}-E)\rho\}$ and type-II error probability $\Tr\{E\sigma\}$.
The definition Eq.~\eqref{eq:def_htre} corresponds to the optimization of the type-II error probability $\Tr\{E\sigma\}$ under the constraint that the type-I error probability has an upper bound $\epsilon$, that is, $\Tr\{(\mathds{1}-E)\rho\}\le\epsilon$.

As a generalization of this definition, we can introduce hypothesis testing relative entropy in GPTs as follows~\footnote{
The previous work~\cite{short2010entropy} does not explicitly define hypothesis testing relative entropy in GPTs like \eqref{eq:def_HTRE_GPT}.
    However, our definition follows the argument in Ref.~\cite{short2010entropy}.
    Indeed, if we take the logarithm of Eq.~(31) in Ref.~\cite{short2010entropy} and minus, we get hypothesis testing relative entropy with $\epsilon=1/2$.
    Our main focuses are properties and applications of the explicit definition of hypothesis testing relative entropy in GPTs.
}:
\begin{definition}[Hypothesis testing relative entropy in GPTs]\label{def:hypo-test-relative-GPTs}
Let $\rho,\sigma\in\Omega$ be states and $q$ be an effect where $0\le \langle q,\rho\rangle\le 1$ holds for any state $\rho\in\Omega$. Let $0\le \epsilon\le 1$ be a real value. We define hypothesis testing relative entropy as follows:
\begin{equation}
    D^\epsilon_{\rH,\rG}(\rho||\sigma):=-\log_2 \min_{\substack{q:\:0\le q\le u,\\\langle q,\rho\rangle\ge 1-\epsilon}}\langle q,\sigma\rangle\;.\label{eq:def_HTRE_GPT}
\end{equation}
\end{definition}

As the following lemma shows, measurement channels do not increase the hypothesis testing relative entropy, which is important for the following discussion.
\begin{lemma}[Data-processing inequality for a measurement channel]\label{lemma:dpi_meas}
    Let $\cE_{\mathbf{e}}:\cS(K)\to S$ defined as $\cE_{\mathbf{e}}(\cdot):=\sum_{j\in J}\langle e_j,\cdot\rangle\ketbra{j}{j}$ be a measurement channel corresponding to the measurement $\mathbf{e}=\{e_j\}_{j\in J}$.
    We have
    \begin{equation}
        D^\epsilon_{\mathrm{H,G}}(\rho||\sigma)\ge D^\epsilon_{\mathrm{H,G}}(\cE_{\mathbf{e}}(\rho)||\cE_{\mathbf{e}}(\sigma))\;.\label{eq:DPI_meas}
    \end{equation}
\end{lemma}

\section{One-shot classical capacity in GPTs}\label{section:capacity}
Here we consider one-shot classical capacity in GPTs based on the setup given by Ref.~\cite{wang2012one}.
First, we describe our setup of one-shot classical information transmission from the sender in the system $\rA$ to the receiver in the system $\rB$.

The sender and receiver share a channel $\Phi$ from  $\cX$ to $\cS(K)$ defined as $\Phi(\ketbra{x}{x})=\sigma^\rB_x$,
where $\cX$ is an alphabet.
The sender encodes an $n$-length bit string $j\in\Gamma:=\{0,1,2,\cdots,2^{n}-1\}$ to $x\in\cX$ by using a function $g(j)=x$ called \textit{encoder}.
Also, the set $\cG=g(\Gamma)$ and the element $g(j)$ are called \textit{codebook} and \textit{codeword}, respectively.
The receiver performs a measurement $\mathbf{m}^\rB:=\{m^\rB_j\}_{j\in\Gamma}$ to the arrived state $\sigma^\rB_{g(j)}$, where $m^\rB_j\ge 0$ and $\sum_{j\in\Gamma}m^\rB_j=u$.
The error probability for a given message $j\in\Gamma$, encoder $g$ and measurement $\mathbf{m}^\rB$ is defined as
\begin{equation}
    \Pr(\mathrm{error}|j,g,\mathbf{m}^\rB)=\langle u- m^\rB_j,\sigma^\rB_{g(j)}\rangle\;.
\end{equation}
This setting is illustrated in Fig.~\ref{fig:setting}.

\begin{figure}[tbp]
    \centering
    \includegraphics[width=\columnwidth]{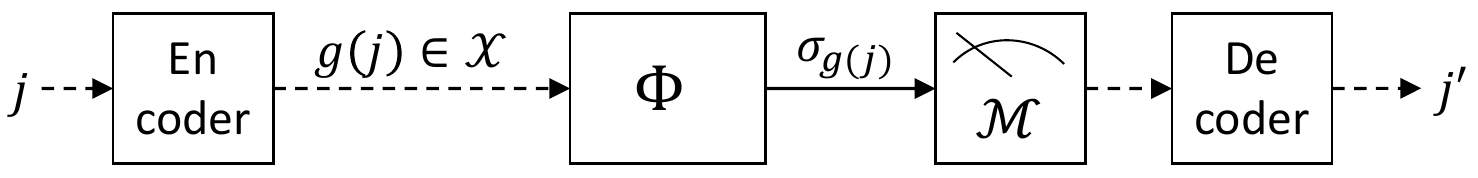}
    \caption{The setup of sending classical information. The dotted arrows mean transmissions of classical information.
    The solid line means a transmission of a state in the general theory.
    The sender chooses the message $j\in\Gamma$ and encodes it by a function $g:\Gamma\to\cX$. 
    Classical information $g(j)$ is transformed into the state $\sigma_{g(j)}$ by a channel $\Phi$ whose input is the classical information and whose output is the state of the GPT of the receiver's system.
    Then the receiver performs a measurement $\cM$ to $\sigma_{g(j)}$ and decodes that result to obtain classical information $j'$. We say that the measurement decodes the message $j$ correctly when $j'=j$.}
    \label{fig:setting}
\end{figure}

The sender and receiver aim to maximize the size of bit strings under the condition that the average error is small enough.
In order to define the rate of this task and the capacity of the channel, we define a code that fulfills this aim.
\begin{definition}[$(2^n,\epsilon)$-code]
    Let $\Gamma=\{0,1,\dots,2^n-1\}$ be a $n$-length bit string.
    A $(2^n,\epsilon)$-code for a map $\Phi:\ketbra{x}{x}\mapsto\sigma^\rB_x$ consists of an encoder $g:\Gamma\to\cX$ and decoding measurement $\mathbf{m}^\rB:=\{m^\rB_j\}_{j\in\Gamma}$ whose average error probability when the messages $j\in\Gamma$ is chosen uniformly at random is bounded from the above by $\epsilon$, in a formula,
    \begin{equation}
        \Pr(\mathrm{error}|g,\mathbf{m}^\rB):=\frac{1}{2^n}\sum_{j\in\Gamma}\Pr(\mathrm{error}|j,g,\mathbf{m}^\rB)\le\epsilon\;.
    \end{equation}
\end{definition}
Then, we define the rate and capacity in the same way as the definition in quantum case~\cite{wang2012one,datta2013smooth}.

\begin{definition}[one-shot $\epsilon$-achievable rate]
    A real number $R\ge 0$ is a one-shot $\epsilon$-achievable rate for one-shot classical information transmission through $\Phi$ if there is a $(2^R,\epsilon)$-code.
\end{definition}
\begin{definition}[one-shot $\epsilon$-classical capacity]
    The one-shot $\epsilon$-classical capacity of a map $\Phi$, $\rC^\epsilon(\Phi)$ is defined as
    \begin{equation}\label{eq:capacity}
       \rC^\epsilon(\Phi):=\sup\{R\mid \mbox{$R$ is a one-shot $\epsilon$-achievable rate}\}\;.
    \end{equation}
\end{definition}

Now, we define the following ensemble:
\begin{align}
    \pi^{\rA\rB}_{P_X}:=\sum_{x\in\cX}P_X(x)\ketbra{x}{x}^\rA\otimes\sigma^\rB_x\;,\label{eq:pi^AB}
\end{align}
where $P_X(x)$ is a probability distribution of a random variable associated with the alphabet $\cX$.
The marginal states with respect to $\rA$ and $\rB$ are 
\begin{align}
    \pi^\rA_{P_X}&=\sum_{x\in\cX}P_X(x)\ketbra{x}{x}^\rA\;,\\
    \pi^\rB_{P_X}&=\sum_{x\in\cX}P_X(x)\sigma^\rB_x\;,
\end{align}
respectively.

In quantum theory, the $\epsilon$-one-shot classical capacity is asymptotically equivalent to the optimal hypothesis testing relative entropy between the above ensemble $\pi^{\rA\rB}$ and the product of its marginal states.
This paper shows that the equivalence also holds even in GPTs.

Firstly, we show the converse part, i.e., the upper bound of $\rC^\epsilon(\Phi)$
by applying Lemma~\ref{lemma:dpi_meas}.
\begin{theorem}\label{theorem:converse}
The $\epsilon$-one-shot classical capacity
of a map $\Phi:\ketbra{x}{x}\mapsto\sigma^\rB_x$ is bounded as follows:
\begin{equation}
    \rC^\epsilon(\Phi)\le \sup_{P_X} D^\epsilon_{\rH,\rG}(\pi^{\rA\rB}_{P_X}||\pi^{\rA}_{P_X}\otimes\pi^{\rB}_{P_X})\;,\label{eq:converse}
\end{equation}
where the supremum is taken over all probability distribution $P_X$,
and where $\pi^{\rA}_{P_X}$ and $\pi^{\rB}_{P_X}$ is a marginal state of $\pi^{\rA\rB}_{P_X}$ with regard to system $\rA$ and $\rB$, respectively.
\end{theorem}
We give a proof of Theorem~\ref{theorem:converse} in Appendix.

Next, we show the achievable part, which gives the lower bound of $\rC^\epsilon(\Phi)$.
\begin{theorem}\label{theorem:achievability}
The $\epsilon$-one-shot classical capacity of a map $\Phi:x\in\cX\to\sigma^\rB_x$ satisfies the following inequality
for any $\epsilon'\in(0,\epsilon)$, any $s>1$ , and any $t>s$ satisfying $\epsilon>s\epsilon'$:
\begin{equation}\label{eq:achievability}
    \rC^\epsilon(\Phi)\ge \sup_{P_X} D^{\epsilon'}_{\rH,\rG}(\pi^{\rA\rB}_{P_X}||\pi^{\rA}_{P_X}\otimes\pi^{\rB}_{P_X})-\log_2\frac{t}{\epsilon-s\epsilon'}\;.
\end{equation}
\end{theorem}

To show Theorem~\ref{theorem:achievability}, we give the following lemma.
\begin{lemma}\label{lemma:good-measurement}
    Let $K$ be a positive cone and let $\cY$ be a finite set.
    Let $\{A_y\}_{y\in \cY}$ be a family of effects in $K^*$ satisfying $0\le A_y\le u$ for all $y\in \cY$.
    Then, for any real numbers $s>1$ and $t>s$, there is a measurement $\{E_y\}_{y\in \cY}\in\cM(K)$ that satisfies the following inequalities for each $E_y$:
    \begin{equation}\label{eq:good-measurement}
        u-E_y\le s(u-A_y)+t\sum_{\substack{z\in \cY\\z\neq y}}A_z\;.
    \end{equation}
\end{lemma}

We give proofs of Theorem~\ref{theorem:achievability} and Lemma~\ref{lemma:good-measurement} in Appendix.
Simply speaking, an achievable code is obtained by Lemma~\ref{lemma:good-measurement} associated with the achievable measurement of hypothesis testing relative entropy.

Here, we remark on the construction of the measurement in the proofs of Lemma~\ref{lemma:good-measurement} and Theorem~\ref{theorem:achievability}.
In classical-quantum channel coding~\cite{holevo1998capacity,schumacher1997sending,wang2012one}, the measurement is chosen as a \textit{pretty good measurement}~\cite{hausladen1994pretty} determined by the square roots of the original family $\{A_y\}_{y\in \cY}$.
On the other hand, our construction is trivial, but it needs the information of the unique index $y_0\in\cY$ satisfying a certain property.
The construction in Refs.~\cite{holevo1998capacity,schumacher1997sending,wang2012one} is valid without such information, and therefore, we can apply the construction in \cite{holevo1998capacity,schumacher1997sending,wang2012one} to actual information tasks in contrast to our construction.
However, this paper aims to estimate the value $\rC^\epsilon(\Phi)$.
For this aim, we only have to show the existence of an optimal measurement, and thus, our construction is sufficient even though our method is not helpful for actual information tasks.

\section{Asymptotic i.i.d. case}\label{section:asymptotic}
In this section, we consider how the capacity is expressed when a channel is used $m$ ($m$ is a positive integer) times in an independent and identical distribution (i.i.d.).

We express $m$-length sequence consists of alphabet $\cX$ as $x_1\dots x_m$.
Let us fix the probability of occurrence for each symbol as $P_X(x)$.
Then, when a channel $\Phi$ is used $m$ times, the sender and receiver can share the following state:
\begin{equation}
\begin{split}
    \pi^{AB}_{P_{X^m}}&:=\sum_{x_1\dots x_m\in\cX^{m}}P_{X^m}(x_1\dots x_m)\ketbra{x_1\dots x_m}{x_1\dots x_m}^\rA\\
    &\qquad\otimes\sigma^\rB_{x_1\dots x_m}\;.
\end{split}
\end{equation}
Here, $\cX^m$ means the set of all $m$-length sequences consist of the alphabet $\cX$ and $\sigma^\rB_{x_1\dots x_m}$ is the abbreviation of $\sigma^\rB_{x_1}\otimes\sigma^\rB_{x_2}\otimes\cdots\otimes\sigma^\rB_{x_m}$.
We denote the above map from $\ketbra{x_1\dots x_m}{x_1\dots x_m}^\rA$ to $\sigma^\rB_{x_1\dots x_m}$ as $\Phi^{\otimes m}$.

Here, we remark on the composition of the model of GPTs.
In the standard setting of GPTs, i.e., the case when we assume no-signaling and local tomography \cite{barrett2007information,janotta2013generalized},
an $n$-composite model of a model defined by $K$ is defined by a positive cone $K_n$ satisfying
\begin{align}\label{eq:composition}
    \bigotimes_{i=1}^n K_i \subset K_n \subset \left(\bigotimes_{i=1}^n K_i^\ast \right)^\ast\;,
\end{align}
where the set $\bigotimes_{i=1}^n K_i$ is defined as
\begin{align}
    \bigotimes_{i=1}^n K_i:=\mathrm{Conv}\{\otimes \rho_i| \rho_i\in K_i\}\;.
\end{align}
As shown in Refs.~\cite{aubrun2021entangleability,aubrun2022entanglement}, an $n$-composite model of a non-classical single model is not uniquely determined in GPTs.
Therefore, we need to be more careful in GPTs when we consider an asymptotic scenario.
However, in the above asymptotic scenario, we only need to consider $m$-uses of a channel $\Phi$, which is a channel from a classical $m$-length bit to an $m$-partite product state $\sigma^\rB_{x_1\dots x_m}$.
Due to the inclusion relation \eqref{eq:composition}, an $m$-partite product state can be regarded as a state in any composite model of a single system,
and therefore, the map $\Phi^{\otimes n}$ is well-defined even in GPTs.
Hence, we can apply the results in the single-shot scenario to the asymptotic scenario.

Now we consider the situation where we encode a message $j\in\Gamma$ to an $m$-length sequence $x_1\dots x_m$ by an encoder $g_m$, that is, $g_m(j)=x_1\dots x_m$.
Here, notice that the size of the set of all messages $\Gamma$ depends on $m$ and we denote it as $|g_m|$.
Also, let us denote the decoding error $\epsilon _m$ when the message $j$ appears uniformly at random.
Then, similar to the single-shot scenario,
we define an $\epsilon$-asymptotic achievable rate as a real number $R\ge 0$ if there exists a sequence of $(m,|g_m|,\epsilon)$-codes satisfying $\liminf_{m\to\infty}\frac{1}{m}\log |g_m|=R$.
Finally, we define $\epsilon$-asymptotic classical capacity following~\cite{verdu1994general}.

\begin{definition}[$\epsilon$-asymptotic classical capacity]
    The $\epsilon$-asymptotic classical capacity of $\Phi$ is defined as follows:
\begin{equation}
    \tilde{\rC}^\epsilon(\Phi):=\sup\left\{R\mid\mbox{$R$ is $\epsilon$-achievable rate for $\Phi$}\right\}\;.
\end{equation}
\end{definition}
By definitions of one-shot $\epsilon$-classical capacity and $\epsilon$-classical capacity, we have
\begin{equation}
    \tilde{\rC}^\epsilon(\Phi)=\liminf_{m\to\infty}\frac{1}{m}\rC^\epsilon(\Phi^{\otimes m})\;.
\end{equation}

Because the map $\Phi^{\otimes m}$ can be regarded as a channel in a model of GPTs,
we can apply Theorem~\ref{theorem:converse}.
Also, in the proof of Theorem~\ref{theorem:achievability},
the decoding measurement is a trivial measurement with information of a given channel.
A trivial measurement is also well-defined in any composite system of the single system,
and therefore, we can also Theorem~\ref{theorem:achievability}.
As a result,
the $\epsilon$-asymptotic classical capacity of $\Phi$ satisfies follows:
\begin{equation}
    \begin{split}
        &\lim_{m\to\infty}\frac{1}{m}\sup_{P_{X^m}} D^{\epsilon'}_{\rH,\rG}(\pi^{\rA\rB}_{P_{X^m}}||\pi^{\rA}_{P_{X^m}}\otimes\pi^{\rB}_{P_{X^m}})\\
        &\le
        \tilde{\rC}^\epsilon(\Phi)\\
        &\le\lim_{m\to\infty}\frac{1}{m}\sup_{P_{X^m}}D^\epsilon_{\rH,\rG}(\pi^{\rA\rB}_{P_{X^m}}\|\pi^\rA_{P_{X^m}}\otimes\pi^{\rB}_{P_{X^m}})\;,
    \end{split}
\end{equation}
where, $\epsilon'\in(0,\epsilon)$.
Here, the hypothesis testing relative entropy $D^{\epsilon}(\rho||\sigma)$ is left-continuous over $\epsilon$ (Lemma~\ref{lemma:continuous} in Appendix).
As a result, we obtain the following theorem.
\begin{theorem}\label{theorem:asymptotic}
    In any model of GPTs and any $\epsilon\in(0,1)$,
    a channel $\Phi$ satisfies
    \begin{equation}
        \tilde{\rC}^\epsilon(\Phi)=\lim_{m\to\infty}\frac{1}{m}\sup_{P_{X^m}}D^\epsilon_{\rH,\rG}(\pi^{\rA\rB}_{P_{X^m}}\|\pi^\rA_{P_{X^m}}\otimes\pi^{\rB}_{P_{X^m}})\;.
    \end{equation}
\end{theorem}
In other words, the classical capacity and the hypothesis testing relative entropy are asymptotically equivalent even in GPTs.

Here, we remark on the dependence of $\epsilon$.
In quantum theory, because of quantum Stein's lemma \cite{hiai1991proper,ogawa2005strong}, we have the following equality \cite{wang2012one}:
\begin{equation*}
    \forall\epsilon\in(0,1),\quad \lim_{n\rightarrow\infty}\frac{1}{n}D^\epsilon_{\rH}(\rho^{\otimes n}||\sigma^{\otimes n})=D(\rho||\sigma)\;.
\end{equation*}
In other words, there is no $\epsilon$-dependence in quantum theory, and the rates are equal to Umegaki relative entropy $D(\rho||\sigma)$~\cite{umegaki1961on,umegaki1962conditional}.
Thus, both asymptotic classical capacity and asymptotic hypothesis testing relative entropy are determined independently with $\epsilon$.
Therefore, we often consider the case $\epsilon\to 0$, which provides Holevo--Schumacher--Westmoreland (HSW) Theorem~\cite{schumacher1997sending,holevo1998capacity} as stated in Ref.~\cite{wang2012one}.
In contrast to quantum theory, it has not been known yet that they are independent with $\epsilon$ in GPTs, which is still an important open problem.
However, whether they are dependent or independent, this paper clarified the asymptotic equivalence between classical capacity and hypothesis testing relative entropy with each $\epsilon$ even in GPTs.

\section{Conclusions}
We have introduced classical information transmission in GPTs, and we have shown the lower and upper bounds of the one-shot classical capacity theorem in any physical theory given by hypothesis testing relative entropy.
Besides, we have shown that the lower and upper bounds are asymptotically equivalent.
In other words, we have shown the equivalence between asymptotic classical capacity and asymptotic hypothesis testing relative entropy in any physical theory.

In classical and quantum theory, the above rates of different information tasks are connected by entropies, but in general, entropies do not possess similar properties to classical and quantum theory in GPTs.
Our contribution is to clarify the equivalence of two rates of different information tasks without entropies even in any physical theory.

Note that the asymptotic equivalence of the capacity and hypothesis testing relative entropy given by Theorem~\ref{theorem:asymptotic} depends on $\epsilon$.
The open problem is then HSW Theorem in GPTs. By virtue of Theorem~\ref{theorem:asymptotic}, this problem can be reduced to the following two problems:
(1) Are asymptotic classical capacity and asymptotic hypothesis testing relative entropy independent with $\epsilon$ even in GPTs?
(2) If so, are asymptotic classical capacity and asymptotic hypothesis testing relative entropy related to standard relative entropy even in GPTs?
The answer to both problems should give an important new operational perspective of entropies and information rates.

\section*{Acknowledgments}
The authors thank Mark M. Wilde for pointing out errors in the previous version and for helpful comments and discussions.
Helpful comments and discussions from Francesco Buscemi are gratefully acknowledged.
Comments from Tan Van Vu on the previous version of the manuscript are gratefully acknowledged.
S.~M. would like to take this opportunity to thank the “Nagoya University Interdisciplinary Frontier Fellowship” supported by Nagoya University and JST, the establishment of university fellowships towards the creation of science technology innovation, Grant Number JPMJFS2120 and “THERS Make New Standards Program for the Next Generation Researchers” supported by JST SPRING, Grant Number
JPMJSP2125.

\bibliography{myref}

\clearpage

\widetext

\appendix

\section*{Lemmas and Proofs}
\subsection{Proof of Lemma~\ref{lemma:dpi_meas}}\label{appendix:dpi}
\begin{proof}
Let $0\le \overline{q}\le u$ be a classical effect that attains the minimization in the definition of classical hypothesis testing relative entropy $D^\epsilon_{\rH,\rG}(\cE_{\mathbf{e}}(\rho)||\cE_{\mathbf{e}}(\sigma))$.
In other words, $\overline{q}$ satisfies the following relations:
\begin{align}
    \langle \overline{q},\cE_{\mathbf{e}}(\rho)\rangle&\ge 1-\epsilon\;,\label{eq:optimal_effect}\\
    -\log_2 \langle\overline{q},\cE_{\mathbf{e}}(\sigma)\rangle&=D^\epsilon_{\mathrm{H,G}}(\cE_{\mathbf{e}}(\rho)||\cE_{\mathbf{e}}(\sigma))\;.\label{eq:optimal_value}
\end{align}
Because of Eq.~\eqref{eq:adjoint}, we have $\langle \cE^\dagger_{\mathbf{e}}(\overline{q}),\rho\rangle=\langle \overline{q},\cE_{\mathbf{e}}(\rho)\rangle$. 
The combination of this with Eq.~\eqref{eq:optimal_effect} and \eqref{eq:optimal_value} leads to the following relations, respectively:
\begin{align}
    \langle \cE^\dagger_{\mathbf{e}}(\overline{q}),\rho\rangle&\ge 1-\epsilon\;,\label{eq:nonoptimal_effect}\\
    -\log_2 \langle\cE^\dagger_{\mathbf{e}}(\overline{q}),\sigma\rangle&=D^\epsilon_{\mathrm{H,G}}(\cE_{\mathbf{e}}(\rho)||\cE_{\mathbf{e}}(\sigma))\;.
\end{align}
Eq.~\eqref{eq:nonoptimal_effect} implies that the effect $\cE^\dagger_{\mathbf{e}}(\overline{q})$ satisfies the condition of the minimization of $D^\epsilon_{\rH,\rG}(\rho||\sigma)$ (but not necessarily the optimal effect that achieves $D^\epsilon_{\rH,\rG}(\rho||\sigma)$), which implies the desired inequality \eqref{eq:DPI_meas}~\footnote{The argument for this proof of the monotonicity is similar to that proposed in Ref.~\cite{wang2012one}, where monotonicity in general CPTP maps is proven in quantum theory. In GPTs, however, we prove the monotonicity only in measurement channels, which is sufficient for this paper's discussion.}.
\end{proof}

\subsection{Proof of Theorem~\ref{theorem:converse}}

\begin{proof}
Let $R:= \rC^\epsilon(\Phi)$ be the maximum $\epsilon$-achievable rate.
Take an $(2^R,\epsilon)$-code, i.e., a encoder $g:\Gamma\to\cX$ and a measurement $\mathbf{m}^\rB:=\{m_j^{\rB}\}_{j\in\Gamma}$ satisfying
\begin{align}\label{eq:proof-theorem1-assumption}
    \frac{1}{2^R}\sum_{j\in\Gamma}\langle u- m^\rB_j,\sigma^\rB_{g(j)}\rangle\le\epsilon\;.
\end{align}
To show the inequality \eqref{eq:converse},
we consider the concrete probability distribution $P_X$ as a uniform distribution,
and we define the bipartite state $\pi^{\rA\rB}_{\mathrm{uni}}$ as
\begin{align}
    \pi^{\rA\rB}_{\mathrm{uni}}:=\frac{1}{2^R}\sum_{j\in\Gamma}\ketbra{g(j)}{g(j)}^\rA\otimes\sigma^\rB_{g(j)}\;.
\end{align}
The marginal states of $\pi^{\rA\rB}_{\mathrm{uni}}$ is given as
\begin{align}
    \pi^{\rA}_{\mathrm{uni}}&=\frac{1}{2^R}\sum_{j\in\Gamma}\ketbra{g(j)}{g(j)}^\rA\;,\\
    \pi^{\rB}_{\mathrm{uni}}&=\frac{1}{2^R}\sum_{j\in\Gamma}\sigma^\rB_{g(j)}\;.
\end{align}
We will show the inequality \eqref{eq:converse} by showing the inequality
\begin{align}\label{eq:proof-theorem1-1}
    R\le D^\epsilon_{\rH,\rG}(\pi^{\rA\rB}_{\mathrm{uni}}||\pi^{\rA}_{\mathrm{uni}}\otimes\pi^{\rB}_{\mathrm{uni}})
\end{align}
as the latter discussion.

Next, take a classical measurement $\{\lambda_{j}^{\rA}\}_{j\in \Gamma}$ such that $\langle\lambda^\rA_{j},\ketbra{g(j')}{g(j')}\rangle=\delta_{j,j'}$.
Then, we define a product measurement $\mathbf{m}':=\{\lambda_j\otimes m_{j'}\}_{j,j'\in\Gamma}$ and the measurement channel $\cE_{\mathbf{m}'}$ associated with the product measurement, i.e., $\cE_{\mathbf{m}'}$ is given as
\begin{align}
    \cE_{\mathbf{m}'}(\rho)=\sum_{j,j'\in\Gamma}\langle\lambda_j\otimes m_{j'},\rho\rangle\ketbra{j}{j}\otimes\ketbra{j'}{j'}\;.
\end{align}
Due to Lemma~\ref{lemma:dpi_meas},
we obtain the following inequality:
\begin{equation}\label{eq:proof-theorem1-2}
   D^\epsilon_{\rH,\rG}\left(\cE_{\mathbf{m}'}\left(\pi^{\rA\rB}_{\mathrm{uni}}\right)||\cE_{\mathbf{m}'}\left(\pi^{\rA}_{\mathrm{uni}}\otimes\pi^{\rB}_{\mathrm{uni}}\right)\right)\le D^\epsilon_{\rH,\rG}(\pi^{\rA\rB}_{\mathrm{uni}}||\pi^{\rA}_{\mathrm{uni}}\otimes\pi^{\rB}_{\mathrm{uni}})\;.
\end{equation}
To calculate the left-hand side of \eqref{eq:proof-theorem1-2},
we define the classical state $P^{\rA\rB}:=\cE_{\mathbf{m}'}(\pi_\mathrm{uni}^{\rA\rB})$.
Then, the marginal states of $P^{\rA\rB}$ are given as
\begin{align}
        P^{\rA}=&\sum_{j,j'\in\Gamma}\langle\lambda_j\otimes m_{j'},\pi^{\rA\rB}_{\mathrm{uni}}\rangle\ketbra{j}{j}^{\rA}\\
        =&\sum_{j,\in\Gamma}\langle\lambda_j\otimes u,\pi^{\rA\rB}_{\mathrm{uni}}\rangle\ketbra{j}{j}^{\rA}\\
        =&\sum_{j,\in\Gamma}\langle\lambda_j,\pi^{\rA}_{\mathrm{uni}}\rangle\ketbra{j}{j}^{\rA}\;,\\
        P^{\rB}=&\sum_{j,j'\in\Gamma}\langle\lambda_j\otimes m_{j'},\pi^{\rA\rB}_{\mathrm{uni}}\rangle\ketbra{j'}{j'}^{\rB}\\
        =&\sum_{j',\in\Gamma}\langle u\otimes m_{j'},\pi^{\rA\rB}_{\mathrm{uni}}\rangle\ketbra{j'}{j'}^{\rB}\\
        =&\sum_{j',\in\Gamma}\langle m_{j'},\pi^{\rB}_{\mathrm{uni}}\rangle\ketbra{j'}{j'}^{\rB}\;.
\end{align}
By these forms of marginal states $P^{\rA}$ and $P^{\rB}$,
we obtain the following relation:
\begin{align}
    P^{\rA}\otimes P^{\rB}=&\sum_{j,j'\in\Gamma} \langle\lambda_j,\pi^{\rA}_{\mathrm{uni}}\rangle\langle m_{j'},\pi^{\rB}_{\mathrm{uni}}\rangle \ketbra{j}{j}^{\rA}\otimes \ketbra{j'}{j'}^{\rB}\\
    =&\sum_{j,j'\in\Gamma} \langle\lambda_j\otimes m_{j'},\pi^{\rA}_{\mathrm{uni}}\otimes\pi^{\rB}_{\mathrm{uni}}\rangle\ketbra{j}{j}^{\rA}\otimes \ketbra{j'}{j'}^{\rB}\\
    =&\cE_{\mathbf{m}'}\left(\pi^{\rA}_{\mathrm{uni}}\otimes\pi^{\rB}_{\mathrm{uni}}\right)\;.
\end{align}
Therefore, to show the desirable inequality \eqref{eq:proof-theorem1-1}, we need to show the inequality
\begin{align}\label{eq:proof-theorem1-3}
    R\le D^\epsilon_{\rH,\rG}(P^{\rA\rB}||P^{\rA}\otimes P^{\rB}),
\end{align}
which will be shown as follows.

To estimate the value $D^\epsilon_{\rH,\rG}(P^{\rA\rB}||P^{\rA}\otimes P^{\rB})$,
take a concrete classical two-outcome measurement $\{q,u-q\}$ to distinguish classical states $P^{\rA\rB}$ and $P^{\rA}\otimes P^{\rB}$ such that $q=(q_i)_{i\in\Gamma^2}\in\mathbb{R}^{\Gamma^2}$ and $q_i=\delta_{j,j'}$ for $i=(j,j')$.
Then, the type-I error probability $\langle q,P^{\rA\rB}\rangle$ is given as
\begin{align}\label{eq:proof-theorem1-4}
    \langle q,P^{\rA\rB}\rangle=&\sum_{j=j'}\langle\lambda_j\otimes m_{j'},\pi_{\mathrm{uni}}^{\rA\rB}\rangle\\
    =&\sum_{j\in\Gamma}\sum_{k\in\Gamma}\frac{1}{2^R}\langle\lambda_j,\ketbra{g(k)}{g(k)}^{\rA}\rangle\langle m_j,\sigma_{g(k)}^{\rB}\rangle\\
    \stackrel{(a)}{=}&\frac{1}{2^R}\sum_{j\in\Gamma}\langle{m_j,\sigma_{g(j)}}\rangle
    \\
    =&1-\frac{1}{2^R}\sum_{j\in\Gamma}\langle u-m_j,\sigma_{g(j)}\rangle\\
    \stackrel{(b)}{\ge} &1-\epsilon\;.
\end{align}
The equation $(a)$ holds
because $\langle \lambda_j,\ketbra{g(k)}{g(k)}=\delta_{j,k}$.
The inequality $(b)$ holds because of the inequality \eqref{eq:proof-theorem1-assumption}.
Also, the type-II error probability $\langle q,P^{\rA}\otimes P^{\rB}\rangle$ is given as
\begin{align}\label{eq:proof-theorem1-5}
    \langle q, P^{\rA}\otimes P^{\rB}\rangle
    =&\sum_{j=j'}\langle\lambda_j\otimes m_{j'},\pi_{\mathrm{uni}}^{\rA}\otimes\pi_{\mathrm{uni}}^{\rA}\rangle\\
    =&\sum_{j\in\Gamma}\langle \lambda_j \pi_{\mathrm{uni}}^{\rA}\rangle \langle m_j \pi_{\mathrm{uni}}^{\rB}\rangle\\
    =&\sum_{j\in\Gamma}\left(\frac{1}{2^R}\sum_{k\in\Gamma}\langle \lambda_j\ketbra{g(k)}{g(k)}^{\rA}\rangle\right)
    \left(\frac{1}{2^R}\sum_{k'\in\Gamma}\langle m_j,\sigma_{k'}^{\rB}\rangle\right)\\
    \stackrel{(a)}{=}&\frac{1}{2^{2R}}\sum_{j\in\Gamma}\sum_{k'\in\Gamma}\langle m_j,\sigma_{k'}^{\rB}\rangle
    =\frac{1}{2^{2R}}\sum_{k'\in\Gamma}\langle u,\sigma_{k'}^{\rB}\rangle
    =\frac{1}{2^R}\;.
\end{align}
The equation $(a)$ holds because $\langle \lambda_j,\ketbra{g(k)}{g(k)}=\delta_{j,k}$.
The relations \eqref{eq:proof-theorem1-4} and \eqref{eq:proof-theorem1-5} imply
the inequality \eqref{eq:proof-theorem1-3}.
As a result, we obtain the inequality \eqref{eq:converse} as follows:
\begin{align}
    \rC^\epsilon(\Phi)=&R\le D^\epsilon_{\rH,\rG}(P^{\rA\rB}||P^{\rA}\otimes P^{\rB})\\
    =&D^\epsilon_{\rH,\rG}\left(\cE_{\mathbf{m}'}\left(\pi^{\rA\rB}_{\mathrm{uni}}\right)||\cE_{\mathbf{m}'}\left(\pi^{\rA}_{\mathrm{uni}}\otimes\pi^{\rB}_{\mathrm{uni}}\right)\right)\\
    \le&D^\epsilon_{\rH,\rG}\left(\pi^{\rA\rB}_{\mathrm{uni}}||\pi^{\rA}_{\mathrm{uni}}\otimes\pi^{\rB}_{\mathrm{uni}}\right)\\
    \le&\sup_{P_X} D^\epsilon_{\rH,\rG}(\pi^{\rA\rB}_{P_X}||\pi^{\rA}_{P_X}\otimes\pi^{\rB}_{P_X})\;.
\end{align}
\end{proof}

Here, we remark that this proof is an operational generalization of the proof of the converse part given in Ref.~\cite{wang2012one}.

\subsection{Proof of Theorem~\ref{theorem:achievability}}

First, we prove Lemma~\ref{lemma:good-measurement}.
\begin{proof}[Proof of Lemma~\ref{lemma:good-measurement}]
    Let us denote $u':=\sum_{y\in \cY}A_y$.
    Also, let us define
    \begin{equation}
        B_y:=-(s-1)u-tu'+(s+t)A_y\;
    \end{equation}
    for each $y\in \cY$.
    Then the desired inequality \eqref{eq:good-measurement} is rewritten as $B_y\le E_y$.

    Next, we show that there exists at most one element $y_0\in \cY$ satisfying $B_{y_0}\ge 0$ by contradiction.
    Assume that there are two different elements $y_0,y_1\in \cY$ satisfying $B_{y_0}\ge0$ and $B_{y_1}\ge0$, respectively.
    Then we have the following two inequalities:
    \begin{align}
        (s+t)A_{y_0}&\ge (s-1)u+tu'\;,\label{eq:proof_lemma_1}\\
        (s+t)A_{y_1}&\ge (s-1)u+tu'\;.\label{eq:proof_lemma_2}
    \end{align}
    Therefore, we have
    \begin{align}
        (s+t)u'&\ge(s+t)(A_{y_0}+A_{y_1})\\
        &\ge 2(s-1)u+2tu'\\
        &>(s+t)u'\;.
    \end{align}
    The first inequality is from $A_{y_0}+A_{y_1}\le u'$.
    The second inequality is from inequalities \eqref{eq:proof_lemma_1} and \eqref{eq:proof_lemma_2}.
    The final inequality holds because $s>1$ and $t>s$.
    However, this is a contradicting relation, and thus, there exists at most one $y_0$ such that $B_{y_0}\ge 0$.

    Finally, we define a measurement.
    Let $y_0\in\cY$ be a unique element satisfying $B_{y_0}\ge0$ if it exists.
    If there does not exist such $y_0$, we choose a fixed element $y_0\in\cY$.
    Then, define a measurement $\{E_y\}_{y\in \cY}$ as
     \begin{align}
        \begin{aligned}\label{def:good-measurement}
            \begin{cases}
                E_y:=0 & (y\neq y_0)\;,\\
                E_y:=u & (y=y_0)\;.
            \end{cases}
        \end{aligned}    
    \end{align}
    As discussed above, $B_i<0$ holds for all $i$ except $y_0$ and therefore, $B_i<E_i$ holds for any $i\neq y_0$,.
    Therefore, what we need to show is $B_{y_0}\le E_{y_0}$, which is shown as follows:
    \begin{align}
        E_{y_0}-B_{y_0}&=u-[-(s-1)u-tu'+(s+t)A_{y_0}]\\
        &=s(u-A_{y_0})+t(u'-A_{y_0})\ge 0\;.
    \end{align}
    As a result, the measurement defined as \eqref{def:good-measurement} satisfies $B_y\le E_y$ for all $y\in \cY$, and thus, the desired inequality \eqref{eq:good-measurement} holds.
\end{proof}

Now, we prove Theorem~\ref{theorem:achievability} by applying Lemma~\ref{lemma:good-measurement}.
\begin{proof}[Proof of Theorem \ref{theorem:achievability}]
    The structure of this proof is also essentially the same as the method given in Ref.~\cite{wang2012one}.

    To show the inequality \eqref{eq:achievability},
    we fix an arbitrary parameter $\epsilon'\in(0,\epsilon)$ and an arbitrary probability distribution $P_X(x)$.
    Then, we take an effect $e$ that achieves the optimization $D^{\epsilon'}_{\rH,\rG}(\pi^{\rA\rB}_{P_X}||\pi^{\rA}_{P_X}\otimes\pi^{\rB}_{P_X})$.
    Therefore, the effect $e$ satisfies $0\le e\le u$ and $\langle e,\pi^{\rA\rB}_{P_X}\rangle\ge 1-\epsilon'$.
    Then, we need to show the existence of $(2^R,\epsilon)$-code satisfying
    \begin{equation}
        R\ge-\log\langle e,\pi^{\rA}_{P_X}\otimes\pi^{\rB}_{P_X}\rangle-\log\frac{t}{\epsilon-s\epsilon'}\;,\label{eq:achievability-desired}
    \end{equation}
    which can be rewritten as
    \begin{equation}
        \epsilon\le s\epsilon'+2^Rt\langle e,\pi^{\rA}_{P_X}\otimes\pi^{\rB}_{P_X}\rangle\;.\label{eq:error-goodcode}
    \end{equation}
    To show this,
    we need to show the existence of the code with the size $R$,
    i.e.,
    the existence of a measurement $\mathbf{m}^B$ such that
    \begin{equation}
        \Pr(\mathrm{error}|\mathbf{m}^B)\le s\epsilon'+2^R t\langle e,\pi^{\rA}_{P_X}\otimes\pi^{\rB}_{P_X}\rangle\;.\label{eq:error-achievability}
    \end{equation}

    Now, we consider the situation to send a $R$-length bit string $j\in\Gamma$.
    We generate an encoder $g:\Gamma\to\cX$ by choosing a codewords $g(j)=x_j\in\cX$ at random, where each $x_j$ is chosen according to the distribution $P_X$ independently.
    Also, let us define $0\le A_x\le u$ as an effect such that
    \begin{equation}
        \langle A_x,\rho^\rB\rangle=\langle e,\ketbra{x}{x}^\rA\otimes\rho^\rB\rangle\label{eq:Ax_e}
    \end{equation}
    for any $\rho^\rB\in\cS(K^\rB)$.
    By applying Lemma~\ref{lemma:good-measurement},
    we choose a decoding measurement $\mathbf{m}^\rB:=\{m^\rB_j\}_{j\in\Gamma}$ satisfying the following inequality holds for any $s>1$ and $t>s$:
    \begin{equation}
        u-m^\rB_j\le s\left(u-A_{g(j)}\right)+t\sum_{\substack{i\in\Gamma\\i\neq j}}A_{g(i)}\;.\label{eq:error_upperbound}
    \end{equation}
    Therefore, an upper bound of
    the error probability swith respect to a message $j\in\Gamma$ with an encoder $g$
    is bounded as follows:
    \begin{equation}
    \mathrm{Pr}(\mathrm{error}|j,g,\mathbf{m}^\rB)\le s\left\langle u-A_{g(j)},\sigma^\rB_{g(j)}\right\rangle+t\sum_{\substack{i\in\Gamma\\i\neq j}}\left\langle A_{g(i)},\sigma^\rB_{g(j)}\right\rangle\;.\label{eq:direct_error}
    \end{equation}

    Next, for an arbitrary fixed bit-string $j$,
    we consider the average of the error probability $\mathrm{Pr}(\mathrm{error}|j,g,\mathbf{m}^\rB)$ over all the encoders $g$.
    We denote the probability to generate $g$ and the average value as $P(g)$ and $\mathrm{Pr}(\mathrm{error}|j,\overline{g},\mathbf{m}^\rB)$, respectively.
    Also, we denote the set of all encoders and the set of encoders satisfying $g(j)=x$ as $G$ and $G_x$, respectively.
    Then, the value $\mathrm{Pr}(\mathrm{error}|j,\overline{g},\mathbf{m}^\rB)$ is bounded as follows due to the way to generate $g$:
    \begin{align}
        \Pr(\mathrm{error}|j,\overline{g},\mathbf{m}^\rB)
        :=&\sum_{g\in G} P(g) \Pr(\mathrm{error}|j,g,\mathbf{m}^\rB)\\
        \stackrel{(a)}{\le} & \sum_{g\in G} P(g) \left(s\left\langle u-A_{g(j)},\sigma^\rB_{g(j)}\right\rangle+t\sum_{\substack{i\in\Gamma\\i\neq j}}\left\langle A_{g(i)},\sigma^\rB_{g(j)}\right\rangle\right)\\
        = & \sum_{x\in\cX} \sum_{g\in G_x} P(g) \left(s\left\langle u-A_{g(j)},\sigma^\rB_{g(j)}\right\rangle+t\sum_{\substack{i\in\Gamma\\i\neq j}}\left\langle A_{g(i)},\sigma^\rB_{g(j)}\right\rangle\right)\\
        \stackrel{(b)}{=} &\sum_{x\in\cX} \sum_{g\in G_x} P(g) s\left\langle u-A_{x},\sigma^\rB_{x}\right\rangle+\sum_{x\in\cX} \sum_{g\in G_x} P(g) t\sum_{\substack{i\in\Gamma\\i\neq j}}\left\langle A_{g(i)},\sigma^\rB_{x}\right\rangle\\
        \stackrel{(c)}{=} &\sum_{x\in\cX} P_X(x) s\left\langle u-A_{x},\sigma^\rB_{x}\right\rangle+\sum_{x\in\cX} \sum_{g\in G_x} P(g) t\sum_{\substack{i\in\Gamma\\i\neq j}}\left\langle A_{g(i)},\sigma^\rB_{x}\right\rangle\\
        \stackrel{(d)}{=} &\sum_{x\in\cX} P_X(x) s\left\langle u-A_{x},\sigma^\rB_{x}\right\rangle+\sum_{x\in\cX} t\sum_{\substack{i\in\Gamma\\i\neq j}}\left\langle P_X(x)\sum_{x'\in\cX} P_X(x')A_{x'},\sigma^\rB_{x}\right\rangle\\
        \stackrel{(e)}{=} & s\left( 1-\sum_{x\in\cX} P_X(x)\left\langle A_{x},\sigma^\rB_{x}\right\rangle\right)+t(2^R-1)\left\langle \sum_{x'\in\cX} P_X(x')A_{x'},\sum_{x\in\cX} P_X(x) \sigma^\rB_{x}\right\rangle\\
        \stackrel{(f)}{\le} & s\left( 1-\sum_{x\in\cX} P_X(x)\left\langle A_{x},\sigma^\rB_{x}\right\rangle\right)+t2^R\left\langle \sum_{x'\in\cX} P_X(x')A_{x'},\sum_{x\in\cX} P_X(x) \sigma^\rB_{x}\right\rangle\;.
    \end{align}
    The inequality $(a)$ is derived from \eqref{eq:direct_error}.
    The equation $(b)$ is derived from the definition of $G_x$.
    The equation $(c)$ holds because the first term is independent with $g$ and the average over $g\in G_x$ is given by $P_X(x)$ due to the way to generate $g$.
    Because each codeword $g(i)$ is chosen at random with probability $P_X$ independent of the choice of $g(j)=x$ in the second term, the average of $A_{g(i)}$ over $g\in G_x$ is given as $P_X(x)\sum_{x'\in\cX}P_X(x')A_{x'}$.
    Therefore, the equation $(d)$ holds.
    The equation $(e)$ holds because the second term is independent of $j$.
    The inequality $(f)$ is shown from the trivial inequality $2^R-1\le 2^R$ and the positivity of the inner product.

    By applying the above upper bound, the average of $\mathrm{Pr}(\mathrm{error}|g,\mathbf{m}^\rB)$ over all of the encoders $g$ is bounded as follows:
    \begin{equation}\label{eq:error_p}
        \begin{split}
            \mathrm{Pr}(\mathrm{error}|\overline{g},\mathbf{m}^\rB):=&\sum_{g\in G} P(g) \sum_{j}\frac{1}{2^R}\mathrm{Pr}(\mathrm{error}|j,g,\mathbf{m}^\rB)\\
            = &\sum_{j}\frac{1}{2^R} \sum_{g\in G} P(g) \mathrm{Pr}(\mathrm{error}|j,g,\mathbf{m}^\rB)\\
            \le & \sum_{j}\frac{1}{2^R} \left(s\left( 1-\sum_{x\in\cX} P_X(x)\left\langle A_{x},\sigma^\rB_{x}\right\rangle\right)+t2^R\left\langle \sum_{x'\in\cX} P_X(x')A_{x'},\sum_{x\in\cX} P_X(x) \sigma^\rB_{x}\right\rangle\right)\\
             \stackrel{(a)}{=} & s\left( 1-\sum_{x\in\cX} P_X(x)\left\langle A_{x},\sigma^\rB_{x}\right\rangle\right)+t2^R\left\langle \sum_{x'\in\cX} P_X(x')A_{x'},\sum_{x\in\cX} P_X(x) \sigma^\rB_{x}\right\rangle\;.
        \end{split}
    \end{equation}
    The equation $(a)$ holds because the summation part is independent of $j$.

    Finally, by using Eq.~\eqref{eq:Ax_e} and the linearity of the inner product, we have
    \begin{equation}\label{eq:pi_AB}
    \begin{split}
        \sum_{x\in\cX} P_X(x)\langle A_x,\sigma^\rB_x\rangle=&\sum_{x\in\cX}P_X(x)\langle e,\ketbra{x}{x}^\rA\otimes\sigma^\rB_x\rangle\\
        =&\left\langle e,\sum_{x\in\cX}P_X(x)\ketbra{x}{x}^\rA\otimes\sigma^\rB_x\right\rangle\\
        =&\langle e,\pi_{P_X}^{\rA\rB}\rangle\ge 1-\epsilon'\;.
    \end{split}
    \end{equation}
    Also, we have
    \begin{equation}\label{eq:piA_times_piB}
    \begin{split}
        \left\langle \sum_{x'}P_X(x')A_{x'},\sum_x P_X(x)\sigma^\rB_x\right\rangle=&\sum_{x'}P_X(x')\left\langle A_{x'}\sum_xP_X(x)\sigma^\rB_x\right\rangle\\
            \stackrel{(a)}{=}&\sum_{x'}P_X(x')\left\langle e,\ketbra{x'}{x'}^\rA\otimes\sum_xP_X(x)\sigma^\rB_x\right\rangle\\
            =&\left\langle e,\sum_{x'}P_X(x')\ketbra{x'}{x'}^\rA\otimes\sum_xP_X(x)\sigma^\rB_x\right\rangle\\
            =&\langle e,\pi^{\rA}_{P_X}\otimes\pi^{\rB}_{P_X}\rangle\;.
    \end{split}
    \end{equation}
    The equality $(a)$ is derived directly from \eqref{eq:Ax_e}.
    Substituting \eqref{eq:pi_AB} and \eqref{eq:piA_times_piB} for \eqref{eq:error_p}, we obtain the inequality
    \begin{equation}
        \Pr(\mathrm{error}|\overline{g},\mathbf{m}^B)\le s\epsilon'+2^Rt\langle e,\pi^{\rA}_{P_X}\otimes\pi^{\rB}_{P_X}\rangle\;.
    \end{equation}
    Due to the same logic as Shannon's random encoding,
    there exists at least one encoder and decoder satisfying the desired inequality \eqref{eq:error-achievability}.
    As a result, we show the existence of the $(2^R,\epsilon)$ code satisfying \eqref{eq:achievability-desired},
    which implies the statement of Theorem~\ref{theorem:achievability}.
\end{proof}

\subsection{Proof of Theorem~\ref{theorem:asymptotic}}

The remaining part of the proof of Theorem~\ref{theorem:asymptotic} is the following lemma.

\begin{lemma}\label{lemma:continuous}
    For any states $\rho,\sigma\in\cS(K)$ and $\epsilon>0$,
    the $\epsilon$-hypothesis testing relative entropy is left-continuous over $\epsilon$.
\end{lemma}

\begin{proof}[Proof of Lemma~\ref{lemma:continuous}]
   Due to the definition \eqref{eq:def_HTRE_GPT},
we only need to show the continuity of the function $D(\epsilon)$ defined as
\begin{align}\label{eq:lemma-continuous}
	D(\epsilon):=\min_{\substack{q:\:0\le q\le u,\\\langle q,\rho\rangle\ge 1-\epsilon}}\langle q,\sigma\rangle
\end{align}
for any $\rho,\sigma\in\cS(K)$.
Because the effect $(1-\epsilon)u$ satisfies the condition of the minimization in \eqref{eq:lemma-continuous}, $0\le D(\epsilon)\le 1-\epsilon$ holds.
Let $\eta>0$ be an arbitrary parameter for the so-called ``$\epsilon-\delta$ discussion".
We only need to show the existence of $\delta$ for any $0<\eta<1-D(\epsilon)$ such that any $\epsilon'<\epsilon$ with $\epsilon-\epsilon'<\delta$ satisfies $|D(\epsilon)-D(\epsilon')|<\eta$.
Take an argument-minimum effect $q_0:=\mathrm{argmin}\;D(\epsilon)$, and we choose $\delta$ as $\delta:=\frac{\epsilon\eta}{1-D(\epsilon)}$.
Besides, we take an effect $q_1$ as
\begin{align}
    q_1:=\frac{\delta}{\epsilon}u+\left(1-\frac{\delta}{\epsilon}\right)q_0\;.
\end{align}
Because $0<\frac{\delta}{\epsilon}\le1$, the effect $q_1$ satisfies $0\le q_1\le u$ and the following inequality:
\begin{align}\label{eq:proof-conti-1}
    \begin{aligned}
        \langle q_1,\rho\rangle=&\left(1-\frac{\delta}{\epsilon}\right)\langle q_0,\rho\rangle + \frac{\delta}{\epsilon}\langle u,\rho\rangle\\
	\ge &\left(1-\frac{\delta}{\epsilon}\right)(1-\epsilon)+\frac{\delta}{\epsilon}\\
	=& 1-\epsilon+\delta\ge 1-\epsilon'\;.
    \end{aligned}
\end{align}
Then, we obtain the following inequality:
\begin{align}
    |D(\epsilon)-D(\epsilon')|\stackrel{(a)}{=}&D(\epsilon')-D(\epsilon)\\
    =&\min_{\substack{q':\:0\le q'\le u,\\\langle q',\rho\rangle\ge 1-\epsilon'}}\langle q',\sigma\rangle-D(\epsilon)\\
    \stackrel{(b)}{\le}&\langle q_1,\sigma\rangle - D(\epsilon)\\
    =&\frac{\delta}{\epsilon}(1-D(\epsilon))=\eta\;.
\end{align}
The equation $(a)$ holds because $\epsilon'<\epsilon$.
The inequality $(b)$ holds because of the inequality \eqref{eq:proof-conti-1}.
As a result, $D(\epsilon)$ is left-continuous, and therefore, $\epsilon$-hypothesis testing relative entropy is also left-continuous over $\epsilon$.
\end{proof}

\end{document}